\documentclass[9pt,a4paper]{article}
\usepackage{epsf, graphicx}
\usepackage{latexsym,amsfonts,amsbsy,amssymb}
\usepackage{amsmath,amsthm}
\usepackage{cite}
\usepackage{cleveref}
\usepackage{indentfirst}
\usepackage{bm}
\makeatletter

\@addtoreset{equation}{section} \makeatother
\newtheorem{theorem}{Theorem}[section]
\newtheorem{lemma}{Lemma}[section]
\newtheorem{corollary}{Corollary}[section]
\newtheorem{remark}{Remark}[section]
\newtheorem{example}{Example}[section]
\newtheorem{proposition}{Proposition}[section]
\newtheorem{definition}{Definition}[section]

\makeatletter \@addtoreset{equation}{section} \makeatother
\newcommand{\udots}{\mathinner{\mskip1mu\raise1pt\vbox{\kern7pt\hbox{.}}
\mskip2mu\raise4pt\hbox{.}\mskip2mu\raise7pt\hbox{.}\mskip1mu}}
\begin{document}

\title{ Non-Invertible-Element Constacyclic  Codes over Finite PIRs\footnote{
 E-Mail addresses: hwliu@mail.ccnu.edu.cn (H. Liu), jinggeliu@mails.ccnu.edu.cn (J. Liu).}}

\author{Hongwei Liu,~Jingge Liu}

\date{\small
School of Mathematics and Statistics, Central China Normal University,Wuhan, 430079, China\\
}
\maketitle


%

\section*{Abstract}
 \addcontentsline{toc}{section}{\protect Abstract} 
 \setcounter{equation}{0} 

In this paper we introduce the notion of $\lambda$-constacyclic codes over finite rings $R$ for arbitary element $\lambda$ of $R$.
We study the non-invertible-element constacyclic codes (NIE-constacyclic codes) over finite principal ideal rings (PIRs).
We determine the algebraic structures of all NIE-constacyclic codes over finite chain rings,  give the unique form of the sets of the defining polynomials and obtain their minimum Hamming distances.
A general form of the duals of NIE-constacyclic codes over finite chain rings is also provided.
In particular, we give a necessary and sufficient condition for the dual of an NIE-constacyclic code to be an NIE-constacyclic code.
Using the Chinese Remainder Theorem, we study the  NIE-constacyclic codes over finite PIRs.
Furthermore, we construct some optimal NIE-constacyclic codes over finite PIRs in the sense that they achieve the maximum possible minimum Hamming distances for some given lengths and cardinalities.

\medskip
\noindent{\large\bf Keywords: }\medskip   finite commutative PIR;  finite commutative chain ring; constacyclic code; minimum Hamming distance.

\medskip
2020 {\it Mathematics Subject Classification.} \, 94B05,~94B60,~94B65

\section{Introduction}

Codes over finite rings  have intrigued a lot of researchers thanks to the discovery that some important families of binary non-linear codes are in fact images under a Gray map of linear codes over $\mathbb{Z}_4$ (see, for example, \cite{rf3,rf4,rf5}).
In particular, the class of constacyclic codes, which contains the well-known class of cyclic and negacyclic codes is interesting for both theoretical and practical aspects.
In the past few decades, scholars have been interested in $\lambda$-constacyclic codes over finite rings where $\lambda$ is invertible.
A constacyclic code is called a \emph{simple-root code} if the characteristic of the finite ring is relatively prime to the length of this code; otherwise it is called a \emph{repeated-root code}.
Dinh and L\'{o}pez-Permouth\cite{rf6} studied the simple-root cyclic codes and negacyclic codes and their duals over a finite chain ring.
As the decomposition of polynomials over finite rings is not unique, the structure of repeated-root constacyclic codes over finite rings is more complex.
Since 2003, some special classes of repeated-root constacyclic codes over certain finite chain rings have been studied by a lot of authors (see, for example,\cite{rf7,rf8,rf9,rf10,rf11,rf12,rf13,rf14,rf15,rf16,rf17,rf18,rf19}).


All previous studies on $\lambda$-constacyclic codes only considered the case when $\lambda$ is invertible.
In this paper, we define the notion of $\lambda$-constacyclic codes where $\lambda$ is non-invertible.
When $\lambda$ is non-invertible, we can determine the algebraic structures of all  $\lambda$-constacyclic codes over finite chain rings and  give a characterization the dual codes of  such $\lambda$-constacyclic codes.
But the minimum Hamming distances of such codes are not good, actually all nonzero  $\lambda$-constacyclic codes have minimum Hamming distance one.
However, when we focus on  $\lambda$-constacyclic codes over finite PIRs, where $\lambda$ is non-invertible, there exist some optimal codes in the sense that they achieve the maximum possible minimum Hamming distances for some given lengths and cardinalities.

Based on this motivation, we generalize the concept of $\lambda$-constacyclic codes over finite rings to the case where $\lambda$ is arbitrary.
When $\lambda$ is non-invertible,  then we call such $\lambda$-constacyclic codes as NIE-constacyclic codes.
We study the NIE-constacyclic codes over finite PIRs.
Firstly, we determine the algebraic structures of all NIE-constacyclic codes over finite chain rings,
obtain the minimum Hamming distance,
and give a general form of the duals of such codes.
In particular, we provide  a necessary and sufficient condition for the dual of an NIE-constacyclic code to be an NIE-constacyclic code.
Moreover, by the Chinese Remainder Theorem, the algebraic structures and  the minimum Hamming distances of NIE-constacyclic codes over finite PIRs can be easily obtained.
It is worth mentioning that we find some optimal codes in the family of NIE-constacyclic codes over finite PIRs.

This paper is organized as follows.
Some necessary background materials are given in Section~$2$.
In  Section $3$, we determine the algebraic structures of all NIE-constacyclic codes over finite chain rings,  give the unique form of the sets of the defining polynomials, obtain the minimum Hamming distance and provide a general form of the duals of such codes.
A necessary and sufficient condition for the dual of an NIE-constacyclic code to be an NIE-constacyclic code is also presented in Section $3$.
We obtain the algebraic structures and the minimum Hamming distances of  NIE-constacyclic codes over finite PIRs and construct some optimal NIE-constacyclic codes over finite PIRs in  Section $4$.

\section{Preliminaries}

\subsection{Finite Chain Rings and Finite PIRs}


Let $\mathcal{R}$ and $\mathcal{R}^{\prime}$ be two finite commutative Frobenius rings, $n$ be a positive integer and $\rho: \mathcal{R}\rightarrow \mathcal{R}^{\prime}$ be a surjective homomorphism.
Then $\rho: \mathcal{R}\rightarrow \mathcal{R}^{\prime}$ can also denote the following three extended maps
$$\rho: \mathcal{R}^n\rightarrow \mathcal{R}^{\prime n}, $$
$$\rho: \mathcal{R}[x]\rightarrow \mathcal{R}^{\prime}[x],$$
$$\rho: \mathcal{R}[x]/ \langle x^n-\lambda\rangle\rightarrow \mathcal{R}^{\prime}[x]/ \langle x^n-\rho(\lambda)\rangle$$ in the usual way.
It is easy to see that $\rho: \mathcal{R}^n\rightarrow \mathcal{R}^{\prime n}$, $\rho: \mathcal{R}[x]\rightarrow \mathcal{R}^{\prime}[x]$ and $\rho: \mathcal{R}[x]/ \langle x^n-\lambda\rangle\rightarrow \mathcal{R}^{\prime}[x]/ \langle x^n-\rho(\lambda)\rangle$ are surjective.

Let $\mathcal{R}$ be a finite Frobenius ring with identity.
A finite ring with identity  is called a \emph{local ring} if it has a unique maximal ideal $M$ and
a \emph{chain ring} if its ideals are linearly ordered by inclusion.
Then the following proposition holds.

\begin{proposition}\label{proo1}(\cite{rf6})
For a finite ring $R$ with identity, the following statements are equivalent:
\begin{description}
  \item[(1)] $R$ is a local ring and the (unique) maximal ideal $M$ of $R$ is principal, i.e., $M=\gamma R$ for some $\gamma \in R$.
  \item[(2)] $R$ is a local PIR.
  \item[(3)] $R$ is a chain ring and all its ideals are given by
  $$\{0\}=\gamma^e R \subsetneqq \gamma^{e-1} R \subsetneqq \cdots \subsetneqq \gamma R \subsetneqq \gamma^{0} R=R,$$
where $e$ is the nilpotency index of $\gamma$.
Moreover,  $R/\gamma R$ is a finite field (called the\emph{ residue field} of $R$) and $| \gamma^{l} R| = |R/\gamma R|^{e-l}$ for $0\leq l \leq e$. (Throughout this paper, $|A|$ denotes the cardinality of the set $A$.)
\end{description}
\end{proposition}

Note that when we take $\gamma=0$, then the nilpotency index of $\gamma$ is $e=1$.
The finite chain ring with the maximal ideal $\{0\}$ is a finite field.

Let $R$ be a finite chain ring with the maximal ideal $\gamma R$ and let
\begin{equation*}
\begin{aligned}
\bar{~}: R  &\longrightarrow  R/\gamma R,\\
                        a   & \longmapsto       a+\gamma R
\end{aligned}
\end{equation*}
be the canonical projection of $R$ onto its residue field.

In the following of this paper,  we use the notation  $R$ to denote the finite commutative chain ring with the maximal ideal $\gamma R$, where $e$ is the nilpotency index of $\gamma$.
Then $R/\gamma R$ is isomorphic to a finite field $\mathbb F_{q}$ for some prime power $q=p^m$, where $p$ is a prime.

\begin{proposition}\label{proo2}(\cite{re2},\cite{re3})
\begin{description}
  \item[(1)] The characteristic of $R$ is $p^l$, where $1\leq l \leq e$.
Moreover, we have $|R|=q^{e}$.
  \item[(2)] There exists an element $\zeta \in R$ having the multiplicative order $q-1$.
Moreover, the cyclic subgroup generated by $\zeta$ is the only subgroup of the unit group of $R$, which is isomorphic to $\mathbb F_{q}\backslash\{0\}$.
$\mathcal{T}_{R}:=\{0,1,\zeta,\zeta^2,\ldots,\zeta^{q-2}\}$ is a complete set of coset representatives modulo $\gamma R$
which is called the \emph{Teichm\"{u}uller set} of $R$.
  \item[(3)] Any element $a\in R$ can be uniquely expressed as
\begin{equation} \label{equu2}
a=a_0+ a_1\gamma+ a_2\gamma^2+\cdots+a_{e-1}\gamma^{e-1},
\end{equation}
where $a_0,a_1,\ldots,a_{e-1}\in \mathcal{T}_{R}$.
Moreover, $a$ is a unit of $R$ if and only if $a_0\neq 0$.
\end{description}
\end{proposition}

For any integer $1\leq j \leq e$,  we define the map $\mu_j$ by
\begin{equation*}
\begin{aligned}
\mu_j: R &\longrightarrow  R/\gamma^jR,\\
                a   & \longmapsto       a+\gamma^jR.
\end{aligned}
\end{equation*}
Then $\mu_j$ is a surjective homomorphism from $R$ to $R/\gamma^jR$.
For any $a\in R$ and  any integer $1\leq j \leq e$, let $a_{\langle j\rangle}$ denote $a+\gamma^jR$ and $R_j$ denote $\mu_j(R)$.
It is obvious that $\mu_1$ is the canonical projection $\bar{~}$ of $R$, $R_1=\overline{R}=R/\gamma R \cong \mathbb F_{q}$ and $\mu_e$ is the identity map of $R$, $R_e=R$.

\begin{theorem}\label{thee1}
 For any integer $1\leq j \leq e$, $R_j=R/\gamma^jR$ is a finite chain ring with the maximal ideal generated by $\gamma_{\langle j\rangle} $ and $j$ is the nilpotency index of $\gamma_{\langle j\rangle}$.
 Moreover, we have the following.
\begin{description}
  \item[(1)] $R_j/\gamma_{\langle j\rangle} R_j \cong \mathbb F_{q}$ and $|R_j|=q^{j}$.
  \item[(2)] $\zeta_{\langle j\rangle}\in R_j$ has the multiplicative order $q-1$.
 The cyclic subgroup generated by $\zeta_{\langle j\rangle}$ is isomorphic to $\mathbb F_{q}\backslash \{0\}$.
 $\mathcal{T}_{R_j}=\{0,~1,\zeta_{\langle j\rangle},\zeta^2_{\langle j\rangle},\ldots,\zeta^{q-2}_{\langle j\rangle}\}$ is a complete set of coset representatives modulo $\gamma_{\langle j\rangle} R_j$.
  \item[(3)] For any $a,~b\in \mathcal{T}_{R_j}$, $a-b$ is $0$ or invertible.
   Any element $b\in R_j$ can be uniquely expressed as
\begin{equation} \label{equu1}
b=b_0+ b_1\gamma_{\langle j\rangle}+ b_2\gamma_{\langle j\rangle}^2+\cdots+b_{j-1}\gamma_{\langle j\rangle}^{j-1},
\end{equation}
where $b_0,b_1,\ldots,b_{j-1}\in \mathcal{T}_{R_j}$.
\end{description}
\end{theorem}

\begin{proof}
It is obvious that $\gamma^j_{\langle j\rangle}=0$.
If $\gamma^k_{\langle j\rangle}=0$ for some nonnegative integer $k$, then $\gamma^k \in \gamma^j R$,
which means that $k\geq j$.
As a result, $j$ is the nilpotency index of $\gamma_{\langle j\rangle}$.

It is easy to see that $\ker(\mu_j)=\gamma^j R$.
Let $S=\{ ~U~|~U~\text{is an ideal of}~R~\text{and}~U\supseteq \gamma^j R~\}$ and
$T=\{ ~V~|~V~\text{is an ideal of}~R/\gamma^jR~\}$.
By the Ideal Correspondence Theorem, the map induced by $\mu_j$ is a bijective map betweens $S$ and $T$.
Note that $S=\{ ~\gamma^k R~|~0\leq k\leq j~\}$.
Then we have $T=\mu_j(S)=\{ ~\gamma_{\langle j\rangle}^k R_j~|~0\leq k\leq j~\}$.
So $\{0\}=\gamma_{\langle j\rangle}^j R_j \subsetneqq \gamma_{\langle j\rangle}^{j-1} R_j \subsetneqq \cdots \subsetneqq \gamma_{\langle j\rangle} R_j \subsetneqq \gamma_{\langle j\rangle}^0 R_j=R_j$ are all ideals of $R_j$, which means that $R_j=R/\gamma^jR$ is a finite chain ring with the maximal ideal generated by $\gamma_{\langle j\rangle}$.

(1). Let $\phi_j$ be the following map
\begin{equation*}
\begin{aligned}
\phi_j :~~ R_j=R/\gamma^j R ~~&\longrightarrow ~~ R_1=R/\gamma R,\\
                   a + \gamma^j R  ~~~~      & \longmapsto  ~~~~ a+\gamma R,
\end{aligned}
\end{equation*}
for any $a\in R$.
It is clear that $\phi_j$ is a well-defined surjective homomorphism and $\ker(\phi_j)=\gamma_{\langle j\rangle} R_j$.
Then $R_j/\gamma_{\langle j\rangle} R_j \cong R/\gamma R \cong \mathbb F_{q}$ and $|R_j|=q^{j}$ by Proposition \ref{proo2}.

(2). Note that $\zeta^{q-1}_{\langle j\rangle}= \left( \mu_j\left(   \zeta \right)\right)^{q-1}= \mu_j\left(   \zeta^{q-1} \right)=\mu_j\left( 1\right) = 1_{\langle j\rangle}$.
If $\zeta^{k}_{\langle j\rangle}= 1_{\langle j\rangle}$ for some $1\leq k < q-1$, then
$\zeta^{k}-1\in \gamma^j R\subseteq \gamma R$.
Thus $\zeta^{k}+\gamma R$ and $1+\gamma R$ are the same coset of $R$ modulo $\gamma R$, which is a contradiction.
Hence, $\zeta_{\langle j\rangle}\in R_j$ has multiplicative order $q-1$ and the cyclic subgroup generated by $\zeta_{\langle j\rangle}$ is isomorphic to $\mathbb F_{q}\backslash\{0\}$.
By Proposition \ref{proo2}, $\{0,~1,\zeta_{\langle j\rangle},\zeta^2_{\langle j\rangle},\ldots,\zeta^{q-2}_{\langle j\rangle}\}$ is a complete set of coset representatives modulo $\gamma_{\langle j\rangle} R_j$.

(3). It follows from Proposition \ref{proo2}.
\end{proof}

\begin{remark}\label{remm2}
For any integer $1\leq j \leq e$, we have proved that
\begin{equation*}
\begin{aligned}
\phi_j :~~ R_j=R/\gamma^j R ~~&\longrightarrow ~~ R_1=R/\gamma R,\\
                   a + \gamma^j R  ~~~~      & \longmapsto  ~~~~ a+\gamma R,~~~~~~~\forall ~a\in R
\end{aligned}
\end{equation*}
 is a well-defined surjective homomorphism and $\ker(\phi_j)=\gamma_{\langle j\rangle} R_j$.
Then $R_j/\gamma_{\langle j\rangle} R_j \cong R/\gamma R \cong \mathbb F_{q}$.
Thus
\begin{equation*}
\begin{aligned}
\Phi_j : ~~~~~~~~~~~~\overline{R_j}=R_j/\gamma_{\langle j\rangle} R_j ~~~&\longrightarrow~ ~\overline{R}= R/\gamma R,\\
     \overline{a_{\langle j\rangle}}= a_{\langle j\rangle} + \gamma_{\langle j\rangle} R_j   ~   & \longmapsto   ~~~\overline{ a}=   a+\gamma R, ~~~~\forall ~a\in R
\end{aligned}
\end{equation*}
is an isomorphism from the residue field of $R_j$ onto the residue field of $R$.
\end{remark}

\subsection{Codes over Finite Rings}

Let $\mathcal{R}$ be a finite Frobenius ring with identity and $n$ be a positive integer.
We call a nonempty subset $C$ of $\mathcal{R}^n$ a \emph{code} of length $n$ over $\mathcal{R}$ and the ring $\mathcal{R}$ is referred to as the \emph{alphabet} of $C$.
If $C$ is an $\mathcal{R}$-submodule of $\mathcal{R}^n$, then $C$ is said to be \emph{linear}.
The \emph{dual} of $C$ is defined as
$$C^{\bot}=\{\textbf{v}\in \mathcal{R}^n |~\langle \textbf{c},\textbf{v}\rangle=0,~\forall ~\textbf{c}\in C\},$$
 where $\langle \textbf{c},\textbf{v}\rangle$ denotes the usual inner product  of $\textbf{c}$ and $\textbf{v}$.

For any codeword $\textbf{c} = (c_0, c_1,\ldots, c_{n-1})$,
define $\text{supp}(\textbf{c})$ to be the set $\{i~|~c_i \neq 0,~0\leq i\leq n-1\}$.
The\emph{ Hamming weight} of $\textbf{c}$ is the cardinality
of the set $\text{supp}(\textbf{c})$ and is denoted by $\text{wt} (\textbf{c})$.
The \emph{Hamming distance} of two codewords $\textbf{a}, \textbf{b}$ is the number of places where they differ, and is denoted by $d (\textbf{a}, \textbf{b})$.
The \emph{minimum Hamming distance }of a linear code $C$,
denoted by $d(C)$, is the minimum Hamming weight of nonzero codewords of $C$.
If $C$ is a zero code, we let $d(C)=n+1$.

For an element $\lambda$ of $\mathcal{R}$, the $\lambda\, $-constacyclic ($\lambda\,$-twisted) shift $\tau _\lambda$ on $\mathcal{R}^n$ is the shift
\begin{displaymath}
\tau_\lambda (\,x_0,\, x_1,\, \cdots , \,x_{n-1})=(\,\lambda x_{n-1},\, x_0,\, x_1, \,\cdots ,\, x_{n-2}).
\end{displaymath}
Let $\tau_\lambda^0 (\textbf{x})=\textbf{x}$ for any $\textbf{x} \in \mathcal{R}^n$.

\begin{definition}\label{def3}
{\rm Let $\lambda$ be any element of $\mathcal{R}$.
 A linear code $C$ is said to be a \emph{$\lambda\, $-constacyclic code} if $\tau_ \lambda (C)\subseteq C$.
When $\lambda$ is non-invertible,  then we call such $\lambda$-constacyclic codes as non-invertible-element constacyclic codes (NIE-constacyclic codes).
}
\end{definition}

\begin{remark}\label{rema3}
The $1$-constacyclic codes are the \emph{cyclic codes} and the $(-1)$-constacyclic codes are just the \emph{negacyclic codes}.
When $\lambda$ is a unit of $\mathcal{R}$, then $\tau_\lambda$ is a bijective $\mathcal{R}$-linear map and a linear code $C$ is $\lambda~ $-constacyclic if and only if $\tau_ \lambda (C)=C$.
When $\lambda$ is non-invertible in $\mathcal{R}$, $\tau_\lambda$ is an $\mathcal{R}$-linear map which is neither injective nor surjective.
\end{remark}

Let $f(x)$ be a polynomial over $\mathcal{R}$ and let $\deg f(x)$ denote the degree of $f(x)$.
Under the standard $\mathcal{R}$-module isomorphism
\begin{displaymath}
\begin{aligned}
\mathcal{R}^n &\longrightarrow \mathcal{R}[x]/\left \langle  x^n-\lambda\right \rangle, \\
(\,c_0,\,c_1,\,\cdots ,\,c_{n-1}\,)& \longmapsto  c_0+c_1x+\cdots+c_{n-1}x^{n-1}+\left \langle  x^n-\lambda\right \rangle,
\end{aligned}
\end{displaymath}
each codeword $\textbf{c}=(\,c_0,\,c_1,\,\cdots ,\,c_{n-1}\,)$ can be  identified with its polynomial representation
$$c(x)=c_0+c_1x+\cdots+c_{n-1}x^{n-1}\in \mathcal{R}\left [ x \right ],\, \deg\,c(x) \leqslant n-1,$$
and each $\lambda$-constacyclic code $C$ of length $n$ over $\mathcal{R}$ can also be viewed as an ideal of the quotient ring $\mathcal{R}[x]/\left \langle  x^n-\lambda\right \rangle.$
In the light of this, the study of $\lambda$-constacyclic codes of length $n$ over $\mathcal{R}$ is equivalent to the study of ideals of the quotient ring $\mathcal{R}[x]/\left \langle  x^n-\lambda\right \rangle.$

In the following of this section, let the notations be as in Subsection 2.1 and $n$ be a positive integer.
Recall that $R$ is a finite commutative chain ring with the maximal ideal $\gamma R$ and $e$ is the nilpotency index of $\gamma$.

\begin{definition}\label{def1}
{\rm Let $C$ be a code of length $n$ over the finite chain ring $R$.
For $0\leq i \leq e-1$, define
$$Tor_i(C)=\{\overline{\textbf{v}}~ | ~\gamma^i\textbf{v} \in C,~\textbf{v} \in R^n\}.$$
$Tor_i(C)$ is called the $i$th \emph{torsion code} of $C$.
$Tor_0(C) =\overline{C}$ is usually called the \emph{residue code} and sometimes is denoted by $Res(C)$.}
\end{definition}

Clearly, $Tor_i(C)$ is a code of length $n$ over the finite field $\overline{R} \cong \mathbb{F}_{q}$.
It is easy to see that $Tor_0(C)\subseteq Tor_1(C) \subseteq \cdots\subseteq Tor_{e-1}(C)$.

\begin{proposition}\label{proo3}(\cite{re4})
For a linear code $C$ over $R$, we have $|C|=\prod_{i=0}^{k-1}|Tor_i(C)|$.
\end{proposition}

\begin{lemma}\label{lemm1}
Let $C$ be a linear code over $R$, then for all $j>i$, we have $Tor_i(C)=\Phi_j\Big( Tor_i \big(\mu_j(C)\big)\Big)$.
\end{lemma}

\begin{proof}
Let $\textbf{v} \in Tor_i \big(\mu_j(C)\big)$, then there exists $\textbf{w}\in R^n$ such that
$\gamma_{\langle j\rangle}^i \textbf{w}_{\langle j\rangle} \in \mu_j(C)$ and
$\overline{\textbf{w}_{\langle j\rangle}}=\textbf{v}$.
Since $\gamma_{\langle j\rangle}^i \textbf{w}_{\langle j\rangle}=(\gamma^i \textbf{w}) _{\langle j\rangle}\in \mu_j(C)$,
there exists $\textbf{z}\in R^n$ such that $\gamma^i \textbf{w}+\gamma^j\textbf{z}\in C$.
Thus $\gamma^i (\textbf{w}+\gamma^{j-i}\textbf{z})\in C$, which leads $\overline{\textbf{w}}\in Tor_i(C)$.
By Remark \ref{remm2}, $\Phi_j(\textbf{v})=\Phi_j(\overline{\textbf{w}_{\langle j\rangle}})=\overline{\textbf{w}}\in Tor_i(C)$.
Thus, $\Phi_j\Big( Tor_i \big(\mu_j(C)\big)\Big)\subseteq Tor_i(C)$.

Conversely, suppose that $\bm{\nu} \in Tor_i (C)$, then there exists $\bm{\omega}\in R^n$ such that
$\gamma^i \bm{\omega}\in C$ and
$\overline{\bm{\omega}}=\bm{\nu}$.
Then $\mu_j(\gamma^i \bm{\omega})= \gamma_{\langle j\rangle}^i \bm{\omega}_{\langle j\rangle}\in \mu_j(C)$,
which implies that $\overline{\bm{\omega}_{\langle j\rangle}} \in Tor_i \big(\mu_j(C)\big)$.
By Remark \ref{remm2}, $\Phi_j^{-1}(\bm{\nu})=\Phi_j^{-1}(\overline{\bm{\omega}})=\overline{\bm{\omega}_{\langle j\rangle}} \in Tor_i \big(\mu_j(C)\big)$.
This means that $\Phi_j^{-1}\big(  Tor_i(C) \big)\subseteq Tor_i \big(\mu_j(C)\big)$.
It follows that $Tor_i(C)\subseteq \Phi_j\Big( Tor_i \big(\mu_j(C)\big)\Big)$.

As a result, $Tor_i(C)=\Phi_j\Big( Tor_i \big(\mu_j(C)\big)\Big)$.
\end{proof}

\section{Constacyclic Codes over Finite Chain Rings}
Throughout this section, let the notions be as in Section 2 and $n$ be a positive integer.
Let $\lambda$ be non-invertible in $R$,
$S:=R[x]/\langle x^n- \lambda \rangle$ and $S_j:=R_j[x]/\langle x^n- \mu_j(\lambda) \rangle$ for $1\leq j \leq e$.
Then each $\mu_j(\lambda)$-constacyclic code of length $n$ over $R_j$ is an ideal of the quotient ring $S_j$ for $1\leq j \leq e$.
Let $\overline{S}:=\overline{R}[x]/\langle x^n\rangle$.
Note that $\overline{S}=S_1$ and  $S_e=S$.

In this section, we determine the algebraic structures of all NIE-constacyclic codes of length $n$ over the finite chain ring $R$.

\subsection{Units in $S$}
Let $a$ be an element of $S=R[x]/\langle x^n- \lambda \rangle$, then $a$ can be uniquely expressed as
\begin{equation} \label{equu3}
a=a_0+a_1 x+ a_2 x^2+\cdots+a_{n-1}x^{n-1},
\end{equation}
where $a_i \in R$ for $0\leq i \leq n-1$.

It is easy to see that the nilpotency index of $x\in S$ is $N:=ne^{\prime}$, where $e^{\prime}$ is the nilpotency index of $\lambda$.

\begin{theorem}\label{thee2}
Let $a=a_0+a_1 x+ a_2 x^2+\cdots+a_{n-1}x^{n-1} \in S$, where $a_i \in R$ for $0\leq i \leq n-1$.
Then $a$ is a unit of $S$ if and only if $a_0$ is a unit of $R$.
Moreover, if $a$ is a unit of $S$, then $a^{-1}=a_0^{-1} \left(1+\sum_{i=1}^{N-1}A^i \right)$, where $A=-a_0^{-1}\left(a-a_0\right)$.
\end{theorem}

\begin{proof}
Suppose that $a$ is a unit of $S$, then there exists $b=b_0+b_1 x+ b_2 x^2+\cdots+b_{n-1}x^{n-1} \in S$ such that $ab=1$.
For convenience, let $a_i=b_i=0$ for $i\geq n$.
Thus
\begin{equation*}
\begin{split}
ab=& \sum_{k=0}^{2n-1}\left( \sum_{i=0}^{k} a_ib_{k-i} \right)x^k \\
    = &\sum_{k=0}^{n-1}\left( \sum_{i=0}^{k} a_ib_{k-i} \right)x^k + \sum_{k=n}^{2n-1}\left( \sum_{i=0}^{k} a_ib_{k-i} \right)x^{n+\left(k-n \right)}\\
    = &\sum_{k=0}^{n-1}\left( \sum_{i=0}^{k} a_ib_{k-i} \right)x^k + \lambda\sum_{k=0}^{n-1}\left( \sum_{i=0}^{k+n} a_ib_{k+n-i} \right)x^{k}\\
     = &\sum_{k=0}^{n-1} \left[  \left( \sum_{i=0}^{k} a_ib_{k-i} \right) + \lambda\left( \sum_{i=0}^{k+n} a_ib_{k+n-i} \right)  \right] x^{k}=1.
\end{split}
\end{equation*}
So $a_0b_0+ \lambda\left( \sum_{i=0}^{n} a_ib_{n-i} \right) =1$ in $R$.
Since $\lambda$ is nilpotent in  $R$, $a_0b_0=1-\lambda\left( \sum_{i=0}^{n} a_ib_{n-i} \right)$ is a unit of  $R$,
which yields that $a_0$ is a unit of $R$.

Conversely, suppose that $a_0$ is a unit of $R$ and let
$$A=-a_0^{-1}(a-a_0)=x\left[-a_0^{-1}(a_1+ a_2 x+\cdots+a_{n-1}x^{n-2})\right].$$
Then $A-1=-a_0^{-1}a$ and $A^N=0$.
Thus $$-1=A^N-1=\left( A-1\right) \left( 1+\sum_{i=1}^{N-1}A^i \right)= -a_0^{-1}\left( 1+\sum_{i=1}^{N-1}A^i \right)a$$ in $S$.
So we have $a_0^{-1}\left( 1+\sum_{i=1}^{N-1}A^i \right)a=1$.
This gives that $a$ is a unit of $S$ and $a^{-1}=a_0^{-1} \left(1+\sum_{i=1}^{N-1}A^i \right)$.
\end{proof}

\begin{theorem}\label{thee3}
Let $a \in S$.
Then $a$ is non-invertible in $S$ if and only if $a\in \langle \gamma, x\rangle$.
$\langle \gamma, x\rangle$ is the unique maximal ideal of $S$ and $S/\langle \gamma, x\rangle \cong \mathbb{F}_{q}$.
Moreover, $\langle \gamma, x\rangle$ is a principal ideal of $S$ if and only if one of the following holds:
\begin{description}
  \item[(i)] $e=1$,
  \item[(ii)] $n=1$,
  \item[(iii)] $e>1$, $n>1$ and $\lambda \in \gamma R\backslash \gamma^2 R$.
\end{description}
\end{theorem}

\begin{proof}
Suppose that $a$ is non-invertible in $S$ and write $a=a_0+a_1 x+ a_2 x^2+\cdots+a_{n-1}x^{n-1}$, where $a_i \in R$ for $0\leq i \leq n-1$.
Then by Theorem \ref{thee2}, we have $a_0 \in R$ is non-invertible.
This means that $a_0\in \gamma R$.
Thus $a\in \langle \gamma, x\rangle$.
Conversely, suppose that $a\in \langle \gamma,x\rangle$, then $a=\gamma b + xc$ for some $b,~c\in S$.
Thus $a^{e+N}=0$, implying that $a$ is non-invertible in $S$.
So we have $\langle \gamma,x\rangle=\{\text{non-invertible elements in } S\}$.
As a result, $\langle \gamma,x\rangle$ is the unique maximal ideal of $S$ and $S/\langle \gamma,x\rangle \cong \mathbb{F}_{q}$.

If $e=1$, then $\gamma=0$ and $R\cong\mathbb{F}_{q}$.
Thus $\langle \gamma,x\rangle=\langle x\rangle$ is a principal ideal.

If $n=1$, then $x=\lambda \in\langle \gamma \rangle$.
Hence $\langle \gamma,x\rangle=\langle\gamma\rangle$ is a principal ideal.

If $e>1$, $n>1$ and $\lambda \in \gamma R\backslash \gamma^2 R$, then $\lambda=\gamma u$ for some unit $u\in R$.
Hence, $\gamma=u^{-1}\lambda=u^{-1}x^n \in \langle x\rangle$.
This means that $\langle \gamma,x\rangle=\langle x\rangle$ is a principal ideal.

If $e>1$, $n>1$ and $\lambda \in \gamma^2 R$ and suppose that $\langle \gamma,x\rangle$ is a principal ideal, then $\langle \gamma,x\rangle= \langle a \rangle$ for some $a \in \langle \gamma,x\rangle$.
So $a$ can be written as $a=a_0+a_1 x+ a_2 x^2+\cdots+a_{n-1}x^{n-1}$,
where $a_0\in \gamma R$ and $a_i \in R$ for $1\leq i \leq n-1$.
Since $\gamma \in \langle a \rangle$, there exists $b=b_0+b_1 x+ b_2 x^2+\cdots+b_{n-1}x^{n-1} \in S$ such that $ab=\gamma$.
For convenience, let $a_i=b_i=0$ for $i\geq n$.
Note that
$$ab=\sum_{k=0}^{n-1} \left[  \left( \sum_{i=0}^{k} a_ib_{k-i} \right) + \lambda\left( \sum_{i=0}^{k+n} a_ib_{k+n-i} \right)  \right] x^{k}=\gamma$$
and we have $$a_0b_0+ \lambda\left( \sum_{i=0}^{n} a_ib_{n-i} \right) =\gamma.$$
If $a_0 \in \gamma^2 R$, by $\lambda \in \gamma^2 R$, then $\gamma=a_0b_0+ \lambda\left( \sum_{i=0}^{n} a_ib_{n-i} \right) \in \gamma^2 R$, which is a contradiction.
Thus $a_0 \notin \gamma^2 R$, which yields that $a_0=\gamma u_0$ for some unit $u_0$ of $R$.
Since $x \in \langle a \rangle$, there exists $c=c_0+c_1 x+ c_2 x^2+\cdots+c_{n-1}x^{n-1} \in S$ such that $ac=x$.
For convenience, let $a_i=c_i=0$ for $i\geq n$.
Note that $$ac=\sum_{k=0}^{n-1} \left[  \left( \sum_{i=0}^{k} a_ic_{k-i} \right) + \lambda\left( \sum_{i=0}^{k+n} a_ic_{k+n-i} \right)  \right] x^{k}=x$$
and we have $$a_0c_0+ \lambda\left( \sum_{i=0}^{n} a_ic_{n-i} \right) =0,~a_0c_1+a_1c_0+ \lambda\left( \sum_{i=0}^{n+1} a_ic_{n+1-i} \right)=1.$$
From $a_0c_0+ \lambda\left( \sum_{i=0}^{n} a_ic_{n-i} \right) =0$, we get $a_0c_0=-\lambda\left( \sum_{i=0}^{n} a_ic_{n-i} \right) \in \gamma^2 R$.
By $a_0=\gamma u_0$, where $u_0$ is a unit of $R$, $c_0\in \gamma R$.
From $a_0c_1+a_1c_0+ \lambda\left( \sum_{i=0}^{n+1} a_ic_{n+1-i} \right)=1$, we get $a_0c_1=1-a_1c_0-\lambda\left( \sum_{i=0}^{n+1} a_ic_{n+1-i} \right)$.
Since $\lambda,~ c_0\in \gamma R$, $a_0c_1$ is invertible in $S$, a contradiction.
As a result, $\langle \gamma,x\rangle$ is a non-principal ideal.
\end{proof}

\begin{corollary}\label{cor1}
\begin{itemize}
  \item[(1)]  If $e=1$, then $S\cong\mathbb{F}_{q}[x]/\langle x^n \rangle$ is a finite chain ring with the maximal ideal $\langle x\rangle$ whose nilpotency index is $n$ and
      $$\{0\}=\langle x^n\rangle\subsetneqq  \langle x^{n-1}\rangle \subsetneqq \cdots \subsetneqq \langle x\rangle \subsetneqq \langle x^0\rangle=S$$ are all ideals of $S$.
      For $0\leq i \leq n$, $|\langle x^i\rangle|=q^{n-i}$.
 \item[(2)]  If $n=1$, then $S\cong R$ is a finite chain ring with the maximal ideal $\langle \gamma\rangle$ whose nilpotency index is $e$ and
      $$\{0\}=\langle \gamma^e\rangle\subsetneqq  \langle \gamma^{e-1}\rangle \subsetneqq \cdots \subsetneqq \langle \gamma\rangle \subsetneqq \langle \gamma^0\rangle=S$$ are all ideals of $S$.
      For $0\leq i \leq e$, $|\langle \gamma^i\rangle|=q^{e-i}$.
  \item[(3)]  If  $e>1$, $n>1$  and $\lambda \in \gamma R\backslash\gamma^2 R$, then $S$  is a finite chain ring with the maximal ideal $\langle x\rangle$ whose nilpotency index is $ne$ and
      $$\{0\}=\langle x^{ne}\rangle\subsetneqq  \langle x^{ne-1}\rangle \subsetneqq \cdots \subsetneqq \langle x\rangle \subsetneqq \langle x^0\rangle=S$$ are all ideals of $S$.
      For $0\leq i \leq ne$, $|\langle x^i\rangle|=q^{ne-i}$.
  \item[(4)] If $e>1$, $n>1$ and $\lambda \in \gamma^2 R$, then $S$  is a finite local ring with the maximal ideal $\langle \gamma,~x\rangle$ but not a chain ring.
\end{itemize}
\end{corollary}

\subsection{Ideals of $S$}
Now we aim to determine the algebraic structures of all $\lambda$-constacyclic codes of length $n$ over $R$ and find a unique representation of ideals in $S:=R[x]/\langle x^n- \lambda \rangle$.

Let $C$ be a $\lambda$-constacyclic code of length $n$ over  $R$, i.e.,  $C$ is an ideal of $S$.
Then for $0\leq i \leq e-1$, $Tor_i(C)$ is an ideal of $\overline{S}=\overline{R}[x]/\langle x^n \rangle\cong \mathbb{F}_{q}[x]/\langle x^n \rangle$, which means that $Tor_i(C)$ is a $0$-constacyclic code of length $n$ over the finite field $\overline{R} \cong \mathbb{F}_{q}$.
By Corollary \ref{cor1},  $Tor_i(C)=\langle x^{T_i} \rangle$ for some $0\leq T_i \leq n$, we say $T_i$ is the \emph{$i$th-torsional degree} of $C$.

Then we can obtain the following result by Definition \ref{def1}, Proposition \ref{proo3} and Corollary \ref{cor1}.

\begin{theorem}\label{thee4}
Let $C$ be an ideal of $S$ and $Tor_i(C)=\langle x^{T_i} \rangle$ for some $0\leq T_i \leq n$.
Then
\begin{description}
  \item[(1)] $|Tor_i(C)|=q^{n-T_i}$.
  \item[(2)] If $g(x) \in S$ and $\gamma^i  \left( x^{t_i}+\gamma g(x) \right) \in C$, then $t_i \geq T_i$.
  \item[(3)] $n\geq T_0 \geq T_1 \geq \cdots \geq T_{e-1} \geq 0$.
  \item[(4)] $|C|=q ^{en- \left( T_0 + T_1 +\cdots + T_{e-1} \right)}$.
\end{description}
\end{theorem}

Let $\mathcal{T}_{R}[x]$ be the set of polynomials in $R[x]$  with coefficients belonging to  $\mathcal{T}_{R}$.
Let $a(x)=\sum_{i=0}^{n-1} a_i x^i\in S$, where $a_i \in R$ for $0 \leq i \leq n-1$.
By Proposition \ref{proo2}, for any $0 \leq i \leq n-1$,
 $a_i\in R$ can be uniquely expressed as
$a_i=a_{0,i}+\gamma a_{1,i}+\gamma^2 a_{2,i}+\cdots+\gamma^{e-1}a_{e-1,i}=\sum_{j=0}^{e-1}\gamma^j a_{j,i}$,
where $a_{j,i}\in \mathcal{T}_{R}$.
Thus, $a(x)=\sum_{j=0}^{e-1}\gamma^j \left(  \sum_{i=0}^{n-1} a_{j,i} x^i \right)$.
For any $0 \leq j \leq e-1$, if $  \sum_{i=0}^{n-1} a_{j,i} x^i = 0$, then let $h_j(x)=0$ and $t_{j}=n-1$, we have $\sum_{i=0}^{n-1} a_{j,i} x^i =x^{t_j} h_j(x)$.
If $  \sum_{i=0}^{n-1} a_{j,i} x^i \neq 0$,
then let $0 \leq t_{j} \leq n-1$ be the smallest integer such that $a_{j,t_{j}}\neq 0$ and  $h_j(x)=\sum_{i=0}^{n-1-t_j} a_{j,i+t_j} x^i$.
Clearly, $h_j(x) \in \mathcal{T}_{R}[x]$ is a unit of $S$ and $\sum_{i=0}^{n-1} a_{j,i} x^i=x^{t_j}\sum_{i=0}^{n-1-t_j} a_{j,i+t_j} x^i=x^{t_j}h_j(x)$.
Thus for any $0 \leq j \leq e-1$, there exists $0 \leq t_{j} \leq n-1$ such that $\sum_{i=0}^{n-1} a_{j,i} x^i =x^{t_j}h_j(x)$, where $h_j(x) \in \mathcal{T}_{R}[x]$ is $0$ or a unit of $S$.
As a result, any polynomial $a(x)\in S$ can be expressed as
\begin{equation} \label{equu4}
a(x)=\sum_{j=0}^{e-1}\gamma^j x^{t_j}h_j(x),
\end{equation}
where $0 \leq t_{j} \leq n-1$  and $h_j(x) \in \mathcal{T}_{R}[x]$ is either zero or a unit of $S$.

It is easy to get the following lemma.
\begin{lemma}\label{lemm2}
\begin{description}
  \item[(1)] For any $1\leq j \leq e$, $\mu_j$ is a bijective map from $\mathcal{T}_{R}[x]$ to $\mathcal{T}_{R_j}[x]$.
  \item[(2)] For any $r(x) \in \mathcal{T}_{R}[x]$, $r(x) = 0$ in $S$ if and only if $\mu_j \big(  r(x) \big)=0$ in $S_j$.
  \item[(3)] For $r(x) \in \mathcal{T}_{R}[x]$ of degree $\leq n-1$, suppose $r(x)\neq 0$ in $S$ and
  we can write $r(x)=x^{n_1}r_1(x)$, $ \mu_j \big(  r(x) \big)=x^{n_2}r_2(x)$ where $0 \leq n_1, n_2 \leq n-1$, $r_1(x)  \in \mathcal{T}_{R}[x]$ is a unit of $S$,
  $r_2(x)  \in \mathcal{T}_{R_j}[x]$ is a unit of $S_j$.
  Then $n_1=n_2$ and $\deg \big( r_1(x) \big) =\deg \big( r_2(x) \big)$.
\end{description}
\end{lemma}

\begin{theorem}\label{thee5}
Let $C$ be an ideal of $S$ and $Tor_i(C)=\langle x^{T_i} \rangle$ for some $0\leq T_i \leq n$.
Then $C$ has the form
\begin{equation} \label{equu5}
C=\langle f_0(x), f_1(x),\ldots, f_{e-1}(x)\rangle,
\end{equation}
such that
\begin{description}
  \item[(i)] when $Tor_i(C)=0$,  $f_i(x)=0$.
  \item[(ii)] when $Tor_i(C)\neq 0$,
\begin{equation} \label{equu6}
f_i(x)=\gamma^i x^{T_i}+ \gamma^{i+1}x^{t_{i+1,i}}h_{i+1,i}(x)+ \gamma^{i+2}x^{t_{i+2,i}}h_{i+2,i}(x)+\cdots + \gamma^{e-1}x^{t_{e-1,i}}h_{e-1,i}(x),
\end{equation}
where $h_{j,i}(x) \in \mathcal{T}_{R}[x]$ is either zero or a unit of $S$ and $t_{j,i} + \deg(h_{j,i}) < T_{j}$.
\end{description}

Moreover, the $e$-tuple $\big(   f_0(x), f_1(x),\ldots, f_{e-1}(x) \big)$ is unique.
\end{theorem}

\begin{proof}
We first prove that $C$ has the form (\ref{equu5}) satisfying (i) and (ii) by induction on the nilpotency index $e$.

When $e=1$, by Corollary \ref{cor1},  $S\cong\mathbb{F}_{q}[x]/\langle x^n \rangle$ and $C=\langle x^{T_0} \rangle$ for some $0\leq T_0 \leq n$.
Thus $Tor_0(C)=C=\langle x^{T_0} \rangle$.
Let
\begin{equation*}
   f_0(x)=
   \begin{cases}
  0,                                          &C=0,\\
  x^{T_0},                              & C \neq 0.
   \end{cases}
  \end{equation*}
Then $C=\langle f_0(x)\rangle$, and $f_0(x)$ satisfies (i) and (ii).
So the statement is true for $e=1$.

Now suppose that any ideal of $S_{e-1}$  has the form given in (\ref{equu5}).
Let $C$ be an ideal of $S$ and $Tor_i(C)=\langle x^{T_i} \rangle$ for some $0\leq T_i \leq n$.
Then $\mu_{e-1}(C)$ is an  ideal of $S_{e-1}$ and hence, by the induction hypothesis, $\mu_{e-1}(C)$ has the form $\langle f_0^{\prime}(x), f_1^{\prime}(x),\ldots, f_{e-2}^{\prime}(x)\rangle.$
By Lemma \ref{lemm1}, for any $0 \leq i \leq e-2$, we have $Tor_i(C)=\Phi_{e-1}\Big( Tor_i \big(\mu_{e-1}(C)\big)\Big)$.
Thus  $Tor_i \big(\mu_{e-1}(C)\big)=\langle x^{T_i} \rangle$.
By the induction hypothesis, for any $0 \leq i \leq e-2$, we have
\begin{itemize}
  \item when $Tor_i \big(\mu_{e-1}(C)\big)=0$,  $f_i^{\prime}(x)=0$.
  \item when $Tor_i \big(\mu_{e-1}(C)\big)\neq 0$,
\begin{equation*}
f_i^{\prime}(x)=\gamma_{\langle e-1\rangle}^i x^{T_i}+ \gamma_{\langle e-1\rangle}^{i+1}x^{t_{i+1,i}^{\prime}}h_{i+1,i}^{\prime}(x)+ \gamma_{\langle e-1\rangle}^{i+2}x^{t_{i+2,i}^{\prime}}h_{i+2,i}^{\prime}(x)+\cdots + \gamma_{\langle e-1\rangle}^{e-2}x^{t_{e-2,i}^{\prime}}h_{e-2,i}^{\prime}(x),
\end{equation*}
where $h_{j,i}^{\prime}(x) \in \mathcal{T}_{R_{e-1}}[x]$ is either zero or a unit of $S_{e-1}$ and $t_{j,i}^{\prime} + \deg(h_{j,i}^{\prime}) < T_{j}$.
\end{itemize}

If $Tor_{e-1}(C)=0$, then $C=0$.
Let $f_i=0$ for all $0 \leq i \leq e-1$.
Then $$C=\langle f_0(x), f_1(x),\ldots, f_{e-1}(x)\rangle$$ and it satisfies (i) and (ii).

If $Tor_{e-1}(C)\neq 0$, then let $f_{e-1}(x)= \gamma^{e-1}x^{T_{e-1}}$.
It is obvious that $f_{e-1}(x) \in C$.
For any $0 \leq i \leq e-2$,
if $Tor_i(C)=0$, then let $f_i(x)=0$.
If $Tor_i(C)\neq 0$,
there exists $F_i(x) \in C$ such that $\mu_{e-1} \big(   F_i(x)  \big)=f_i^{\prime}(x)$.
According to (\ref{equu4}), we write
$$F_i(x)=\sum_{j=0}^{e-2}\gamma^j x^{t_{j,i}}h_{j,i}(x)+\gamma^{e-1}H_{e-1,i}(x),$$
where $0 \leq t_{j,i} \leq n$  and $h_{j,i} \in \mathcal{T}_{R}[x]$ is either zero or a unit of $S$ and $H_{e-1,i}(x)\in\mathcal{T}_{R}[x]$.
Thus we have
$$\mu_{e-1} \big(   F_i(x)  \big)=\sum_{j=0}^{e-2}\gamma_{\langle e-1\rangle}^j x^{t_{j,i}}\mu_{e-1} \big( h_{j,i}(x)\big)=
\gamma_{\langle e-1\rangle}^i x^{T_i}+\sum_{j=i+1}^{e-2} \gamma_{\langle e-1\rangle}^{j}x^{t_{j,i}^{\prime}}h_{j,i}^{\prime}(x).$$
It follows that
1) for $0\leq j \leq i-1$, $h_{j,i}(x)=0$;
2) $t_{i,i}=T_i$ and $h_{i,i}(x)=1$;
3) for $i+1 \leq j \leq e-2$,  $t_{j,i}=t_{j,i}^{\prime}$ and $\mu_{e-1} \big( h_{j,i}(x)\big)=h_{j,i}^{\prime}(x)$.
Hence
$$F_i(x)=\gamma^i x^{T_i}+\sum_{j=i+1}^{e-2} \gamma^{j}x^{t_{j,i}}h_{j,i}(x)+\gamma^{e-1}H_{e-1,i}(x)$$
and $t_{j,i} + \deg(h_{j,i}) = t_{j,i}^{\prime} + \deg(h_{j,i}^{\prime}) < T_{j}$.
Let
$$H_{e-1,i}(x)=\sum_{k=0}^{n-1}z_k x^k,~\widetilde{H}_{e-1,i}(x)=\sum_{k=T_{e-1}}^{n-1}z_k x^{k-T_{e-1}},$$ where $z_k \in \mathcal{T}_{R}$.
Then
$$H_{e-1,i}(x)-\widetilde{H}_{e-1,i}(x)x^{T_{e-1}}=\sum_{k=0}^{T_{e-1}-1}z_k x^k$$  can be written as $x^{t_{e-1,i}}h_{e-1,i}(x)$, where $h_{e-1,i} \in \mathcal{T}_{R}[x]$ is either zero or a unit of $S$ and $t_{e-1,i}+ \deg (h_{e-1,i}) < T_{e-1}$.
Let
$$f_i(x)=\gamma^i x^{T_i}+ \sum_{j=i+1}^{e-2} \gamma^{j}x^{t_{j,i}}h_{j,i}(x)+ \gamma^{e-1}x^{t_{e-1,i}}h_{e-1,i}(x).$$
We can get that
$$f_i(x)=F_i(x)-\widetilde{H}_{e-1,i}(x)f_{e-1}(x) \in C$$
and $f_i(x)$ satisfies (i) and (ii).
As a result, $f_0(x), f_1(x),\ldots, f_{e-1}(x) \in  C $ and satisfy (i) and (ii).

We claim that $C=\langle f_0(x), f_1(x),\ldots, f_{e-1}(x)\rangle$.
First of all, since $f_0(x), f_1(x),\ldots, f_{e-1}(x) \in  C $, we have $\langle f_0(x), f_1(x),\ldots, f_{e-1}(x)\rangle \subseteq C$.
Conversely, suppose $c(x) \in C$, then $\mu_{e-1} \big( c(x) \big) \in \mu_{e-1}(C)$ and hence
$$\mu_{e-1} \big( c(x) \big)=\sum_{i=0}^{e-2}a_i^{\prime}(x)f_i^{\prime}(x),$$ where $a_i^{\prime}(x)\in S_{e-1}$.
Let $a_i(x) \in S_{e}$ such that $\mu_{e-1} \big( a_i(x) \big)=a_i^{\prime}(x) $ for $0 \leq i \leq e-2$.
Thus $$\mu_{e-1} \big( c(x) \big)=\mu_{e-1} \big( \sum_{i=0}^{e-2}a_i(x)f_i(x) \big),$$
which means that $$c(x)= \sum_{i=0}^{e-2}a_i(x)f_i(x)+\gamma^{e-1}x^ta_{e-1}(x),$$ for some $0\leq t \leq n-1$ and $a_{e-1}(x)\in \mathcal{T}_{R}[x]$ which is either zero or a unit of $S$.
It follows that $\gamma^{e-1}x^ta_{e-1}(x) \in C$.
If $a_{e-1}(x)=0$, then  $c(x)= \sum_{i=0}^{e-2}a_i(x)f_i(x) \in \langle f_0(x), f_1(x),\ldots, f_{e-1}(x)\rangle$.
If $a_{e-1}(x) \in \mathcal{T}_{R}[x]$ is a unit of $S$, then $\gamma^{e-1}x^t \in C$ and hence
 $x^t \in Tor_{e-1}(C)= \langle x^{T_{e-1}}\rangle$.
This implies that $t\geq T_{e-1}$ and so we have
$$c(x)= \sum_{i=0}^{e-2}a_i(x)f_i(x)+x^{t-T_{e-1}}a_{e-1}(x)f_{e-1}(x) \in \langle f_0(x), f_1(x),\ldots, f_{e-1}(x)\rangle.$$
As a result, $C\subseteq\langle f_0(x), f_1(x),\ldots, f_{e-1}(x)\rangle$.
Thus, we have shown that $C=\langle f_0(x), f_1(x),\ldots, f_{e-1}(x)\rangle$ as claimed.

To prove the uniqueness, we suppose that $C=\langle g_0(x), g_1(x),\ldots, g_{e-1}(x)\rangle$
such that
\begin{description}
  \item[(i)] when $Tor_i(C)=0$,  $g_i(x)=0$.
  \item[(ii)] when $Tor_i(C)\neq 0$,
\begin{equation*}
g_i(x)=\gamma^i x^{T_i}+ \gamma^{i+1}x^{s_{i+1,i}}w_{i+1,i}(x)+ \gamma^{i+2}x^{s_{i+2,i}}w_{i+2,i}(x)+\cdots + \gamma^{e-1}x^{s_{e-1,i}}w_{e-1,i}(x),
\end{equation*}
where $w_{j,i}(x) \in \mathcal{T}_{R}[x]$ is either zero or a unit of $S$ and $s_{j,i} + \deg(w_{j,i}) < T_{j}$.
\end{description}
If $C=0$, then $g_i(x)=0$ for all $0 \leq i\leq e-1$.
Thus $$\left(   f_0(x), f_1(x),\ldots, f_{e-1}(x) \right)=\left(  g_0(x), g_1(x),\ldots, g_{e-1}(x) \right) = (0,0,\ldots, 0).$$
If $C \neq 0$, then $Tor_{e-1}(C)\neq 0$ and it is clear that $g_{e-1}=f_{e-1}=\gamma^{e-1} x^{T_{e-1}}$.
Consider $$g_{e-2}-f_{e-2}=\gamma^{e-1} \left( x^{s_{e-1,e-2}}w_{e-1,e-2}(x)-x^{t_{e-1,e-2}}h_{e-1,e-2}(x) \right) \in C$$
which can be written as $\gamma^{e-1}x^k h(x)$,
 where $0\leq k \leq n-1$ and $h(x)\in \mathcal{T}_{R}[x]$ is either zero or a unit of $S$.
If $h(x)$ is a unit of $S$, then $\gamma^{e-1}x^k \in C$ and $k \leq T_{e-1}-1 < T_{e-1}$.
$\gamma^{e-1}x^k \in C$  implies $x^k \in Tor_{e-1}(C)= \langle  x^{T_{e-1}} \rangle$, which means $k \geq T_{e-1}$,
a contradiction.
Therefore, $h(x)=0$ and so $g_{e-2}=f_{e-2}$.
Proceeding inductively, we have that $g_{i}=f_{i}$ for all $0 \leq i\leq e-1$.
Thus $$\big(   f_0(x), f_1(x),\ldots, f_{e-1}(x) \big)=\big(  g_0(x), g_1(x),\ldots, g_{e-1}(x) \big),$$
which means the expression is unique.
\end{proof}

\begin{definition}\label{def2}
{\rm Let $C$ be an ideal of $S$.
 We define the unique $e$-tuple obtained from Theorem \ref{thee5} to be the \emph{representation} of $C$.
 In that case, we also say that $C=\langle\langle f_0(x), f_1(x),\ldots, f_{e-1}(x)\rangle\rangle$.}
\end{definition}

\begin{example}\label{exa1}
If  $e>1$, $n>1$ and $\lambda \in \gamma R\backslash \gamma^2 R$,
it is shown that in Corollary \ref{cor1},
$$\{0\}=\langle x^{ne}\rangle\subsetneqq  \langle x^{ne-1}\rangle \subsetneqq \cdots \subsetneqq \langle x\rangle \subsetneqq \langle x^0\rangle=S$$ are all ideals of $S$.
Let $C$ be a nonzero ideal of  $S$, then there exists $0\leq j \leq ne-1$ such that $C=\langle x^j \rangle$.
There also exist $0 \leq k \leq e-1$ and $0 \leq w \leq n-1$ such that $j=kn+w$ and so $C=\langle x^{kn+w} \rangle =\langle \gamma^kx^w\rangle$.
Note that for $ i > k$, $\gamma^i=(\gamma^kx^w)x^{(i-k)n-w}u$ for some unit $u$ of $R$.
It is easy to see that
\begin{equation*}
   Tor_i(C)=
   \begin{cases}
  0,                                          & i < k,\\
  \langle x^w \rangle,        & i =k,\\
  \mathbb{F}_{q}^n,    & i > k,
   \end{cases}
  \end{equation*}
and
\begin{equation*}
   T_i=
   \begin{cases}
  n,                                          & i < k,\\
  w,        & i =k,\\
  0,    & i > k.
   \end{cases}
  \end{equation*}
By Theorem \ref{thee5}, $C=\langle x^j \rangle=\langle\langle0,\ldots, 0, \gamma^k x^w,\gamma^{k+1},\ldots, \gamma^{e-1}\rangle\rangle$.
\end{example}

According to the proof of the uniqueness in Theorem \ref{thee5}, we can easily obtain the following.

\begin{corollary}\label{cor2}
Let $C=\langle\langle f_0(x), f_1(x),\ldots, f_{e-1}(x)\rangle\rangle$ be an ideal of $S$.
Assume that $Tor_i(C)\neq 0$.
Then $f_i(x)$ is the unique polynomial in $C$ which has the form (\ref{equu6}).
\end{corollary}

In the following, we obtain the minimum Hamming distances of all nonzero NIE--constacyclic codes over finite chain rings.

\begin{theorem}\label{thee6}
Let $C$ be a nonzero $\lambda$-constacyclic code of length $n$ over  $R$.
Then $d(C)=1$.
\end{theorem}

\begin{proof}
Suppose $\lambda=0$, let $\textbf{c} = (c_0, c_1,\ldots, c_{n-1})$ be a nonzero codeword of $C$.
Let $t=\min\{  i ~|~ c_i \neq 0, 0\leq i \leq n-1 \}$.
Then $\tau_\lambda^{n-t-1}(\textbf{c})=(0,\ldots,0,c_t) \in C$.
So we have $d(C)=1$.

Suppose $\lambda\neq 0$ is non-invertible.
Then $\lambda=\gamma^ku$, where $1\leq k \leq e-1$ and $u$ is a unit of $R$.
Let $\textbf{c} = (c_0, c_1,\ldots, c_{n-1})$ be a nonzero codeword of $C$ and write each $c_i$ as $c_i=\gamma^{k_i}u_i$, where $1\leq k_i \leq e-1$ and $u_i$ is zero or a unit of $R$.

(1) If $|\{  k_i ~|~ u_i \neq 0, 0\leq i \leq n-1 \}|=1$, let $l=\min\{  k_i ~|~ u_i \neq 0, 0\leq i \leq n-1 \}$.
Then $\textbf{c} =\gamma^l (u_0, u_1,\ldots, u_{n-1})$ and
$\gamma^{e-l-1}\textbf{c}=\gamma^{e-1}(u_0, u_1,\ldots, u_{n-1}) \in C$ is a nonzero codeword.
Let $t=\min\{  i ~|~ u_i \neq 0, 0\leq i \leq n-1 \}$.
Thus $\tau_\lambda^{n-t-1}(\gamma^{e-l-1}\textbf{c})=(0,\ldots,0,\gamma^{e-1}u_t) \in C$ and
$\text{wt} \left(\tau_\lambda^{n-t-1}(\gamma^{e-l-1}\textbf{c}) \right)=1$.
This implies  $d(C)=1$.

(2) If $|\{  k_i ~|~ u_i \neq 0, 0\leq i \leq n-1 \}|\geq 2$, let $l^{\prime}=\min\{  k_i ~|~ u_i \neq 0, 0\leq i \leq n-1 \}$ and
$t=\min\{  i ~|~ k_i = l^{\prime}, 0\leq i \leq n-1 \}$.
Then $\gamma^{e-l^{\prime}-1}\textbf{c}=\gamma^{e-1}(u_0^{\prime}, u_1^{\prime},\ldots, u_{n-1}^{\prime}) \in C$,
where
\begin{equation*}
   u_i^{\prime}=
   \begin{cases}
  0,                                          & u_i=0,\\
  0,                                          & u_i\neq 0 ~\text{and}~ k_i\neq l^{\prime},\\
  u_i ,                                     & u_i\neq 0  ~\text{and}~k_i = l^{\prime}.
   \end{cases}
  \end{equation*}
So $\tau_\lambda^{n-t-1}(\gamma^{e-l^{\prime}-1}\textbf{c})=(0,\ldots,0,\gamma^{e-1}u_t) \in C$.
Note that $\gamma^{e-1}u_t\neq 0$ and we have
$$\text{wt} \left(\tau_\lambda^{n-t-1}(\gamma^{e-l^{\prime}-1}\textbf{c}) \right)=1.$$
Hence, $d(C)=1$.
\end{proof}

\subsection{The Dual Codes}
Let $\mathcal{R}$ be a finite Frobenius ring with identity, $\widehat{\lambda}$ be an element in $\mathcal{R}$.
For any $\mathbf{a}=(a_0,a_1,\cdots ,a_{n-1}) \in \mathcal{R}^n$, define $\mathcal{P}(\mathbf{a}):=\sum_{i=0}^{n-1}a_ix^i \in \mathcal{R}[x]$.

For a $\widehat{\lambda}$-constacyclic code $C$ of length $n$ over $\mathcal{R}$, define $\mathcal{A}(C)$ to be $\{\mathbf{a} \in \mathcal{R}^n~|~\mathcal{P}(\mathbf{c})\mathcal{P}(\mathbf{a})=0~\text{in}~ \mathcal{R}[x]/ \langle  x^n-\widehat{\lambda} \rangle,~ \text{for any} ~ \mathbf{c} \in C\}$.
It is obvious that  $\mathcal{A}(C)$  is also a $\widehat{\lambda}$-constacyclic code of length $n$ over $\mathcal{R}$.

For any integer $k$, let $P_k=\begin{pmatrix}
&   &  &  1\\
&   & \udots & \\
& 1 &  & \\
1 & & &
\end{pmatrix}_{k\times k}$.
Let $\pi$ be given by \begin{equation*}
\begin{aligned}
\pi:~~~\mathcal{R}^n ~~&\longrightarrow~~\mathcal{R}^n, \\
\mathbf{c}=(c_0,c_1,\dots ,c_{n-2},c_{n-1})~~&\longmapsto ~~\mathbf{c}P_n=(c_{n-1},c_{n-2},\dots , c_1,c_0).
\end{aligned}
\end{equation*}
It is easy to see that $\pi$ is a permutation of coordinates and
 $\pi^2$ is the identity map on $\mathcal{R}^n$.

We have the following.

\begin{theorem}\label{thee9}
 Let $C$ be a  $\widehat{\lambda}$-constacyclic code of length $n$ over  $\mathcal{R}$.
Then $C^{\perp}=\pi \big(\mathcal{A}(C) \big).$
\end{theorem}

\begin{proof}
We claim that for $\mathbf{a} \in \mathcal{R}^n$, $\mathbf{a} \in C^{\perp}$ if and only if $\pi(\mathbf{a})\in \mathcal{A}(C)$, i.e., $\mathcal{P} (\mathbf{c})\mathcal{P} (\pi(\mathbf{a}))=0$ in $\mathcal{R}[x]/ \langle  x^n-\widehat{\lambda} \rangle$ for any $\mathbf{c} \in C$.
Thus $C^{\perp}=\pi^{-1} \big(\mathcal{A}(C) \big)=\pi \big(\mathcal{A}(C) \big)$.

Now we prove the above claim.
We write $\mathbf{a}=(a_0,a_1,\cdots ,a_{n-1}),~\mathbf{c}=(c_0,c_1,\cdots ,c_{n-1})$.
Let $a_j=a'_{n-1-j}$ and then $\mathcal{P} (\pi(\mathbf{a}))=\sum_{i=0}^{n-1}a'_{i}x^i$.
Let $a'_i=c_i=0$ for all $i\geq n$.

Suppose that $\mathbf{a} \in C^{\perp}$, then $\langle \tau^k_{\widehat{\lambda}}( \mathbf{c}), \mathbf{a}  \rangle =0$ for any $\mathbf{c} \in C$ and $0\leq k \leq n-1$.
Note that $\tau^k_{\widehat{\lambda}}( \mathbf{c})=(\widehat{\lambda} c_{n-k},\cdots,\widehat{\lambda} c_{n-1}, c_0, \cdots , c_{n-k-1})$.
Then we have  $\langle \tau^k_{\widehat{\lambda}}( \mathbf{c}), \mathbf{a}  \rangle = \sum_{i=0}^{n-k-1}c_ia_{k+i}+\widehat{\lambda}\sum_{i=n-k}^{n-1}c_ia_{k+i-n} = 0$.
Hence,
\begin{equation*}
\begin{split}
\mathcal{P} (\mathbf{c})\mathcal{P} (\pi(\mathbf{a}))
    = &\sum_{k=0}^{n-1}\left( \sum_{i=0}^{k} c_ia'_{k-i} \right)x^k + \sum_{k=n}^{2n-2}\left( \sum_{i=0}^{k} c_ia'_{k-i} \right)x^{n+\left(k-n \right)}\\
    = &\sum_{k=0}^{n-1}\left( \sum_{i=0}^{k} c_ia'_{k-i} \right)x^k + \widehat{\lambda} \sum_{k=n}^{2n-2}\left( \sum_{i=k-n+1}^{n-1} c_ia'_{k-i} \right)x^{k-n}\\
    = &\sum_{k=0}^{n-1}\left( \sum_{i=0}^{k} c_ia'_{k-i} \right)x^k + \widehat{\lambda}\sum_{k=0}^{n-2}\left( \sum_{i=k+1}^{n-1} c_ia'_{k+n-i} \right)x^{k}\\
     = &\sum_{k=0}^{n-2}  \left( \sum_{i=0}^{k} c_ia'_{k-i} + \widehat{\lambda}  \sum_{i=k+1}^{n-1}c_ia'_{k+n-i}  \right) x^{k}+ \left(\sum_{i=0}^{n-1}c_ia'_{n-1-i} \right)x^{n-1}\\
     =&\sum_{k=0}^{n-2}  \left( \sum_{i=0}^{k} c_ia_{n-1-k+i} + \widehat{\lambda}  \sum_{i=k+1}^{n-1}c_ia_{i-k-1}  \right) x^{k}+ \left(\sum_{i=0}^{n-1}c_ia_{i} \right)x^{n-1}\\
     =&\sum_{k=0}^{n-2} \langle \tau^{n-1-k}_{\widehat{\lambda}}( \mathbf{c}), \mathbf{a}  \rangle x^{k}+ \langle \mathbf{c}, \mathbf{a}\rangle x^{n-1}\\
     =&0.
\end{split}
\end{equation*}

Conversely, suppose $\mathcal{P} (\mathbf{c})\mathcal{P} (\pi(\mathbf{a}))=0$ in $\mathcal{R}[x]/ \langle  x^n-\widehat{\lambda} \rangle$ for any $\mathbf{c} \in C$.
Since $\mathcal{P} (\mathbf{c})\mathcal{P} (\pi(\mathbf{a}))=\sum_{k=0}^{n-2} \langle \tau^{n-1-k}_{\widehat{\lambda}}( \mathbf{c}), \mathbf{a}  \rangle x^{k}+ \langle \mathbf{c}, \mathbf{a}\rangle x^{n-1}=0$,
$\langle \mathbf{c}, \mathbf{a}\rangle=0$.
This yields that $\mathbf{a} \in C^{\perp}$.
We have proved the claim.
The result then follows.
\end{proof}

Recall that $R$ is a finite commutative chain ring 
and $\lambda$ is a non-invertible element in $R$.

For a $\lambda$-constacyclic code $C$ of length $n$ over $R$,  $\mathcal{A}(C)$  is also a $\lambda$-constacyclic code $C$ of length $n$ over $R$.
In the light of this, $C$ and $\mathcal{A}(C)$ will also be viewed as ideals of $\mathcal{R}[x]/ \langle  x^n-\lambda \rangle$.

\begin{remark}\label{rema4}
Let $C=\langle\langle f_0(x), f_1(x),\ldots, f_{e-1}(x) \rangle\rangle \subseteq \mathcal{R}[x]/ \langle  x^n-\lambda \rangle$ be a $\lambda$-constacyclic code of length $n$ over  $R$.
Since $\mathcal{A}(C)$  is also  a $\lambda$-constacyclic code of length $n$ over $R$, we can assume that $\mathcal{A}(C)$ has the form $\mathcal{A}(C)=\langle\langle g_0(x), g_1(x),\ldots, g_{e-1}(x) \rangle\rangle \subseteq\mathcal{R}[x]/ \langle  x^n-\lambda \rangle$.
In general, it is not easy to determine all $g_i(x)$s.
But in some special cases like `` $e\leq 2$" or `` $e>2$ and $f_i(x)=0$ for all $0\leq i \leq e-3$", $\mathcal{A}(C)$ can be easily obtained.
\end{remark}

\begin{theorem}\label{thee10}
Let $C\neq R^n$ be a $\lambda$-constacyclic code of length $n$ over  $R$.
Then $d(C^{\perp})=1$.
\end{theorem}

\begin{proof}
Since $C^{\perp}=\pi \big(\mathcal{A}(C) \big)$, $d(C^{\perp})=d\big(\mathcal{A}(C) \big)$.
It is easy to see that if $C\neq R^n$, then $\mathcal{A}(C) \neq 0$.
Since $\mathcal{A}(C)$ is a  $\lambda$-constacyclic code of length $n$ over  $R$, we have $d\big(\mathcal{A}(C) \big)=1$ by Theorem~\ref{thee6}.
Thus $d(C^{\perp})=1$.
\end{proof}

If $\widehat{\lambda}$  is a unit of $R$, it is known that the dual code of a $\widehat{\lambda}$-constacyclic code is a $\widehat{\lambda}^{-1}$-constacyclic code.
When $\widehat{\lambda}$ is non-invertible in $R$, the following examples show that the dual code of a  $\widehat{\lambda}$-constacyclic code may not be a constacyclic code.

\begin{example}\label{exa2}
Let $e=1$, then $R\cong\mathbb{F}_{q}$.
It is shown in Corollary \ref{cor1} that
$$\{0\}=\langle x^{n}\rangle\subsetneqq  \langle x^{n-1}\rangle \subsetneqq \cdots \subsetneqq \langle x\rangle \subsetneqq \langle x^0\rangle=S$$ are all ideals of $S$.
Let $C$ be a proper ideal of  $S$, i.e., $C\neq \{0\}$ and $C\neq S$, then there exists $1\leq i \leq n-1$ such that $C=\langle x^i \rangle$.
Then
$$G_1=\left(
  \begin{array}{ccccccccc}
    O_{(n-i)\times i} & I_{(n-i)\times (n-i)} \\
  \end{array}
\right)$$
is a generator matrix of $C$,
where $O_{k\times t}$ denotes the $k\times t$ zero matrix and $I_{t \times t}$ denotes the $t \times t$  identity matrix.

It is easy to see that
$$H_1=\left(
  \begin{array}{ccccccccc}
   I_{ i\times i}  & O_{i \times(n-i)}  \\
  \end{array}
\right)$$
 is a generator matrix of $C^{\perp}$.
\end{example}

For any $\widehat{\lambda} \in R$, note that $\textbf{a}=(\underbrace{0,\ldots,0}_{i-1},1,\underbrace{0,\ldots,0}_{n-i}) \in C^{\perp}$ and
$\tau_{\widehat{\lambda} }(\textbf{a})=(\underbrace{0,\ldots,0}_{i},1,\underbrace{0,\ldots,0}_{n-i-1}) \notin C^{\perp}$.
It follows that $C^{\perp}$ is not a $\widehat{\lambda}$-constacyclic code for any  $\widehat{\lambda} \in R$.

\begin{example}\label{exa3}
Let $e>1$, $n>0$ and $\lambda \in \gamma R\backslash \gamma^2 R$.
By  Example \ref{exa1}, any proper ideal of  $S$ has the form
$C=\langle\langle0,\ldots, 0, \gamma^k x^w,\gamma^{k+1},\ldots, \gamma^{e-1}\rangle\rangle$ for some $0 \leq k \leq e-1$ and $0 \leq w \leq n-1$ and $w,~k$ are not all zero.
Then a generator matrix of $C$ is
$$G_2=\left(
  \begin{array}{ccccccccc}
    O_{(n-w)\times w} & \gamma^kI_{(n-w)\times (n-w)} \\
    \gamma^{k+1}I_{ w\times w}  & O_{w \times(n-w)}
  \end{array}
\right).$$
It is easy to see that
$$H_2=\left(
  \begin{array}{ccccccccc}
   \gamma^{e-(k+1)}I_{ w\times w}  & O_{w \times(n-w)}  \\
   O_{(n-w)\times w} & \gamma^{e-k}I_{(n-w)\times (n-w)}
  \end{array}
\right)$$
 is a generator matrix of $C^{\perp}$.

If $w=0$, then $C^{\perp}=\langle \gamma^{e-k}\rangle$ is a $\widehat{\lambda}$-constacyclic code for any  $\widehat{\lambda} \in R$.
If $w\neq 0$, for any $\widehat{\lambda} \in R$, note that $\textbf{b}=(\underbrace{0,\ldots,0}_{w-1}, \gamma^{e-(k+1)},\underbrace{0,\ldots,0}_{n-w}) \in C^{\perp}$ and
$\tau_{\widehat{\lambda} }(\textbf{b})=(\underbrace{0,\ldots,0}_{w}, \gamma^{e-(k+1)},\underbrace{0,\ldots,0}_{n-w-1}) \notin C^{\perp}$.
It follows that $C^{\perp}$ is not a $\widehat{\lambda}$-constacyclic code for any  $\widehat{\lambda} \in R$.
\end{example}

Next, we will give a necessary and sufficient condition for the dual of a $\lambda$-constacyclic code to be a constacyclic code.
Firstly, we need the following lemma.

\begin{lemma}\label{lemm3}
Let $C$ be a $\lambda$-constacyclic code of length $n$ over  $R$ and $Tor_i(C)=\langle x^{T_i} \rangle$, where $0\leq T_i \leq n$, for $i=0,1,\ldots,e-1$.
Then $Tor_i \big(\mathcal{A}(C)\big)=\langle x^{n-T_{e-1-i}} \rangle$ for $i=0,1,\ldots,e-1$.
\end{lemma}

\begin{proof}
For any $i=0,1,\ldots,e-1$, let $Tor_i\big(\mathcal{A}(C)\big)=\langle x^{W_i} \rangle$, for some $0\leq W_i \leq n$.
Then there exists $f(x)\in S$ such that $\gamma^ix^{W_i}+\gamma^{i+1}f(x)\in \mathcal{A}(C)$.
Since $Tor_{e-1-i}(C)=\langle x^{T_{e-1-i}} \rangle$,
there exists $g(x)\in S$ such that $\gamma^{e-1-i}x^{T_{e-1-i}}+\gamma^{e-i}g(x)\in C$.
Thus $$\big(\gamma^ix^{W_i}+\gamma^{i+1}f(x)\big)\big(\gamma^{e-1-i}x^{T_{e-1-i}}+\gamma^{e-i}g(x)\big)= \gamma^{e-1}x^{W_i+T_{e-1-i}}=0,$$
which means that $W_i+T_{e-1-i}\geq n$, i.e., $W_i\geq n-T_{e-1-i}$.

Note that $$\left|\mathcal{A}(C) \right|=\left| C^{\perp} \right|= \dfrac{|R|^n}{|C|}= \dfrac{p^{men}}{p^{m\sum_{i=0}^{e-1}(n-T_i)}}=p^{m\sum_{i=0}^{e-1}T_i}$$ and
$$\left|\mathcal{A}(C) \right|= \prod_{i=0}^{e-1}\left| Tor_i\big(\mathcal{A}(C)\big)\right|=  p^{m\sum_{i=0}^{e-1}(n-W_i)}.$$
It follows that $\sum_{i=0}^{e-1}T_i=\sum_{i=0}^{e-1}(n-W_i)$, i.e., $\sum_{i=0}^{e-1}W_i=\sum_{i=0}^{e-1}(n-T_i)$.
By  $W_i\geq n-T_{e-1-i}$, we have $W_i= n-T_{e-1-i}$, i.e., $Tor_i \big(\mathcal{A}(C)\big)=\langle x^{n-T_{e-1-i}} \rangle$.
\end{proof}

Recall for any integer $k$, $P_k=\begin{pmatrix}
&   &  &  1\\
&   & \udots & \\
& 1 &  & \\
1 & & &
\end{pmatrix}_{k\times k}$.
For any $i\times j$ matrix $A$, let $A^{\circ}:=P_iAP_j.$

\begin{theorem}\label{thee66}
Let the notions be as in Theorem \ref{thee5}.
Let $C$ be a $\lambda$-constacyclic code of length $n$ over  $R$.
Then

\begin{description}
  \item[(1)] $C^{\perp}$ is permutation equivalent to a $\lambda$-constacyclic code.
  \item[(2)] $C^{\perp}$ is a $\widehat{\lambda}$-constacyclic code for some $\widehat{\lambda} \in R$ if and only if $C=\gamma^iR^n$ for some $0\leq i \leq e$.
      Moreover, if $C=\gamma^iR^n$ for some $0\leq i \leq e$, then
$C^{\perp}=\gamma^{e-i}R^n$ and $C,~C^{\perp}$ are $\widehat{\lambda}$-constacyclic codes for any  $\widehat{\lambda} \in R$.
\end{description}

\end{theorem}

\begin{proof}
(1)~Since $\mathcal{A}(C)$  is a $\lambda$-constacyclic code of length $n$ over $R$ and $C^{\perp}=\pi \big(\mathcal{A}(C) \big)$, $C^{\perp}$ is permutation equivalent to a $\lambda$-constacyclic code.

(2)~When $C=R^n$, then $C^{\perp}=0$ is a constacyclic code.

Next, let $C\neq R^n$.
From Lemma~\ref{lemm3}, $Tor_i \big(\mathcal{A}(C)\big)=\langle x^{n-T_{e-1-i}} \rangle$ for $i=0,1,\ldots,e-1$.
Let $W_i= n-T_{e-1-i}$ for $i=0,1,\ldots,e-1$.
Since $C\neq R^n$, $\mathcal{A}(C)\neq 0$.
Let $0\leq i_1 \leq e-1$ be the smallest integer $i$ such that $W_i\neq n$
and let $0 \leq i_1 < i_2 \cdots < i_k \leq e-1$, where $1\leq k \leq e$ such that
$n>W_{i_1}=W_{i_1+1}=\cdots = W_{i_2-1}> W_{i_2}=W_{i_2+1}=\cdots = W_{i_3-1}>\cdots>W_{i_{k-1}}=W_{i_{k-1}+1}= \cdots = W_{i_{k}-1}> W_{i_k} =W_{i_k+1} = \cdots = W_{e-1} \geq 0.$

Since $\mathcal{A}(C)$ is a $\lambda$-constacyclic code of length $n$ over  $R$,
$\mathcal{A}(C)$ has the form
\begin{equation*}
\mathcal{A}(C)=\langle g_0(x), g_1(x),\ldots, g_{e-1}(x)\rangle \subseteq \mathcal{R}[x]/ \langle  x^n-\lambda \rangle,
\end{equation*}
such that
\begin{description}
  \item[(i)] for $0\leq i < i_1$, $g_i(x)=0$.
  \item[(ii)] for $i_1\leq i \leq e-1$, $g_i(x)=\gamma^i x^{W_i}+ \sum_{j=i+1}^{e-1} \gamma^{j}x^{w_{j,i}}h_{j,i}(x),$
where $h_{j,i}(x) \in \mathcal{T}_{R}[x]$ is either zero or a unit of $S$ and $w_{j,i} + \deg(h_{j,i}) < W_{j}$.
\end{description}
Then
$$\begin{pmatrix}
\gamma^{i_1+1}A_{1,1}  &  \gamma^{i_1+1}A_{1,2}  &  \gamma^{i_1+1}A_{1,3}  &   \cdots &  \gamma^{i_1+1}A_{1,k}                & \gamma^{i_1}I_{n-W_{i_1}}\\
\gamma^{i_2+1}A_{2,1} &  \gamma^{i_2+1}A_{2,2} &  \gamma^{i_2+1}A_{2,3} &  \cdots & \gamma^{i_2}I_{W_{i_1}-W_{i_2}}  &                                                 \\
\vdots                               &    \vdots                             &   \vdots                               &  \udots &                                                             &                                                 \\
\gamma^{i_{k-1}+1}A_{k-1,1}&\gamma^{i_{k-1}+1}A_{k-1,2}&\gamma^{i_{k-1}}I_{W_{i_{k-2}}-W_{i_{k-1}}}&  &  &\\
\gamma^{i_k+1}A_{k,1}&\gamma^{i_{k}}I_{W_{i_{k-1}}-W_{i_{k}}}  & & &  & \\
\end{pmatrix}$$
is a generator matrix of $\mathcal{A}(C) $, where $I_a$ is the $a \times a$ identity matrix for any $a$.
By the proof of (1), we have $C^{\perp}=\pi \big(\mathcal{A}(C) \big)$, then
$$\begin{pmatrix}
\gamma^{i_1}I_{n-W_{i_1}} &  \gamma^{i_1+1}A^{\circ}_{1,k}  & \cdots & \gamma^{i_1+1}A^{\circ}_{1,3} & \gamma^{i_1+1}A^{\circ}_{1,2} & \gamma^{i_1+1}A^{\circ}_{1,1}\\
                         &   \gamma^{i_2}I_{W_{i_1}-W_{i_2}} &  \cdots  &  \gamma^{i_2+1}A^{\circ}_{2,3} & \gamma^{i_2+1}A^{\circ}_{2,2}   &  \gamma^{i_2+1}A^{\circ}_{2,1} \\
                          &                                                        &  \ddots    & \vdots                              &      \vdots                                 &         \vdots                           \\
&             &         & \gamma^{i_{k-1}}I_{W_{i_{k-2}}-W_{i_{k-1}}} & \gamma^{i_{k-1}+1}A^{\circ}_{k-1,2}  & \gamma^{i_{k-1}+1}A^{\circ}_{k-1,1}\\
&  & & &  \gamma^{i_{k}}I_{W_{i_{k-1}}-W_{i_{k}}} & \gamma^{i_k+1}A^{\circ}_{k,1}\\
\end{pmatrix}$$
is a generator matrix of $C^{\perp}$.

Assume that $C^{\perp}$ is a $\widehat{\lambda}$-constacyclic codes for some $\widehat{\lambda} \in R$.
Suppose that $k> 1$, then there exists $$\textbf{a}=(\underbrace{0,\ldots,0,\gamma^{i_1}}_{n-W_{i_1}},a_{n-W_{i_1}},a_{n-W_{i_1}+1},\ldots, a_{n-1}) \in C^{\perp}.$$
Note that $$\tau_{\widehat{\lambda} }(\textbf{a})= (\underbrace{\widehat{\lambda}a_{n-1},0,\ldots,0}_{n-W_{i_1}}, \gamma^{i_1}, a_{n-W_{i_1}}, \ldots, a_{n-2})  \notin C^{\perp},$$
which is a contradiction.
Hence $k=1$.
Suppose that $W_{i_1}\neq0,$ then $\begin{pmatrix}
\gamma^{i_1}I_{n-W_{i_1}} &  \gamma^{i_1+1}A^{\circ}_{1,1} \\
\end{pmatrix}$ is a generator matrix of $C^{\perp}$, where $A^{\circ}_{1,1}$ is an $W_{i_1}\times n$ matrix.
So there exists $$\textbf{b}=(\underbrace{0,\ldots,0,\gamma^{i_1}}_{n-W_{i_1}},b_{n-W_{i_1}},b_{n-W_{i_1}+1},\ldots, b_{n-1}) \in C^{\perp}.$$
But $$\tau_{\widehat{\lambda} }(\textbf{b})= (\underbrace{\widehat{\lambda}b_{n-1},0,\ldots,0}_{n-W_{i_1}}, \gamma^{i_1}, b_{n-W_{i_1}}, \ldots, b_{n-2})  \notin C^{\perp},$$ which is a contradiction.
Thus $W_{i_1}=0$.
As a result, $\gamma^{i_1}I_{n}$ is a generator matrix of $C^{\perp}$, which yields that $C^{\perp}=\gamma^{i_1}R^n$, $C=\gamma^{e-i_1}R^n$ and $1\leq e-i_1 \leq e$.

On the other hand, if $C=\gamma^iR^n$ for some $0\leq i \leq e$, then
$C^{\perp}=\gamma^{e-i}R^n$.
It is obvious that $C,~C^{\perp}$ are $\widehat{\lambda}$-constacyclic codes for any  $\widehat{\lambda} \in R$.
\end{proof}

\section{Constacyclic Codes over Finite PIRs}
In this section, $\mathbf{R}$ is always a finite PIR with identity.
Then $\mathbf{R}$ is isomorphic to a product of finite chain rings, which means that there exists a ring isomorphism
\begin{equation*}
\begin{aligned}
\psi : ~~\mathbf{R}~~ & \longrightarrow ~~R^{(1)} \times R^{(2)} \times \cdots \times R^{(s)} \\
                   r ~~  & \longmapsto ~~  (r^{(1)},r^{(2)},\ldots, r^{(s)}),
\end{aligned}
\end{equation*}
where $R^{(t)}$ is a finite commutative chain ring with identity and $r^{(t)}\in R^{(t)}$ for $1\leq t \leq s.$
Let
\begin{equation*}
\begin{aligned}
\psi^{(t)} : ~~\mathbf{R}~~ & \longrightarrow ~~R^{(t)}\\
                   r ~~  & \longmapsto ~~  r^{(t)},
\end{aligned}
\end{equation*}
for each $1\leq t \leq s$ and $\psi(r)=(r^{(1)},r^{(2)},\ldots, r^{(s)})=\left( \psi^{(1)}(r),\psi^{(2)}(r),\ldots, \psi^{(s)}(r) \right)$.
It is clear that $\psi^{(t)}$ is a surjective homomorphism.
As we mentioned in Section 1, $\psi^{(t)}: \mathbf{R}\rightarrow R^{(t)}$ can be extended to the following maps
$$\psi^{(t)}: \mathbf{R}^n\rightarrow R^{(t)n}, $$
$$\psi^{(t)}: \mathbf{R}[x]\rightarrow R^{(t)}[x],$$
$$\psi^{(t)}: \dfrac{\mathbf{R}[x]}{ \langle x^n-\lambda\rangle}\rightarrow \dfrac{R^{(t)}[x]}{ \langle x^n-\psi^{(t)}(\lambda)\rangle}$$ in the usual way.

Let $\lambda \in \mathbf{R}$ be a non-invertible element and $\psi(\lambda)=\left( \psi^{(1)}(\lambda),\psi^{(2)}(\lambda),\ldots, \psi^{(s)}(\lambda) \right)$.
$\Psi$ denotes the following map
\begin{equation}\label{equu7}
\begin{aligned}
\Psi : ~~~~\frac{\mathbf{R}[x]}{\langle x^n-\lambda \rangle}~~ & \longrightarrow ~~ \frac{R^{(1)}[x]}{\langle x^n-\psi^{(1)}(\lambda) \rangle} \times  \frac{R^{(2)}[x]}{\langle x^n-\psi^{(2)}(\lambda) \rangle} \times \cdots \times  \frac{R^{(s)}[x]}{\langle x^n-\psi^{(s)}(\lambda) \rangle},\\
f(x)           ~~  & \longmapsto ~~  \left(\psi^{(1)}\left(  f(x)\right),\psi^{(2)}\left(  f(x)\right),\ldots, \psi^{(s)}\left(  f(x)\right)\right).
\end{aligned}
\end{equation}
It is easy to see that $\Psi $ is a surjective isomorphism.

Let $C^{(t)}$ be an ideal of $\dfrac{R^{(t)}[x]}{\langle x^n-\psi^{(t)}(\lambda) \rangle}$ for each $1\leq t \leq s$. Then \emph{the Chinese
product} is defined by
$\text{CRT}(C^{(1)}, C^{(2)},\cdots, C^{(s)})=\{ ~\Psi^{-1}(f^{(1)}, f^{(2)},\cdots, f^{(s)})~ |~ f^{(t)} \in C^{(t)}, ~t=1,2,\ldots,s~\}.$

From (\ref{equu7}), $C$ is a  $\lambda$-constacyclic code of length $n$ over  $\mathbf{R}$ if and only if for any $~1\leq t \leq s,$ $\psi^{(t)}(C)$ is a  $\psi^{(t)}(\lambda)$-constacyclic code of length $n$ over  $R^{(t)}$.
Thus, the following theorem can be easily obtained.

\begin{theorem}\label{thee7}
$C$ is an ideal of $\dfrac{\mathbf{R}[x]}{\langle x^n-\lambda \rangle}$ if and only if  $\psi^{(t)}(C)$ is an ideal of $\dfrac{R^{(t)}[x]}{\langle x^n-\psi^{(t)}(\lambda) \rangle}$ for each $1\leq t \leq s$.

If $C^{(t)}$ is an ideal of $\dfrac{R^{(t)}[x]}{\langle x^n-\psi^{(t)}(\lambda) \rangle}$ for each $1\leq t \leq s$, then $\text{CRT}(C^{(1)}, C^{(2)},\cdots, C^{(s)})$ is an ideal  of  $\dfrac{\mathbf{R}[x]}{\langle x^n-\lambda \rangle}$.
\end{theorem}

According to the reference \cite{rff8}, we can obtain the minimum Hamming distances of the NIE-constacyclic codes over finite PIRs as follows.
\begin{theorem}\label{thee8}(\cite{rff8})
Let $C$ be a nonzero $\lambda$-constacyclic code of length $n$ over  $\mathbf{R}$.
Then $ d(C) = \min \{ ~d\left(\psi^{(t)}(C) \right)~|~ 1\leq t \leq s \}$.
\end{theorem}




Combining  Theorem \ref{thee6} with Theorem \ref{thee8}, we can easily obtain the following.

\begin{corollary}\label{cor3}
Let $C$ be a $\lambda$-constacyclic code of length $n$ over  $\mathbf{R}$.
Assume that there exists $1 \leq j  \leq s$ such that $\psi^{(j)}(\lambda)$ is a non-invertible element of $R^{(j)}$ and $\psi^{(j)}(C)\neq 0$.
Then  $ d(C) =1$.
\end{corollary}

\begin{remark}\label{rema5}
{\rm Let $\mathbf{R}=\psi^{-1}(\underbrace{R^{0} \times R^{0} \times \cdots \times R^{0}}_{s})$, where $R^{0}$ is a finite chain ring,
then $|\mathbf{R}|=|R^{0}|^s$.
Let $C^{0}$ be an MDS $\lambda^0$-constacyclic code of length $n$ over $R^{0}$, i.e., $d(C^{0})=n+1-\log_{|R^{0}|}|C^{0}|$.
Let $C=\text{CRT}(\underbrace{C^{0} \times C^{0} \times \cdots C^{0}}_{s-1} \times 0), ~\lambda=\psi^{-1}(\underbrace{\lambda^{0},\lambda^{0},\ldots, \lambda^{0}}_{s-1}, 0)\in \mathbf{R}$,
then $C$ is a $\lambda$-constacyclic code  of length $n$ over $\mathbf{R}$ and $|C|=|C^{0}|^{s-1}$.
By Theorem \ref{thee8}, we have $d(C)=d(C^{0})$.
On the other hand, $n+1-\log_{|\mathbf{R}|}|C|=n+1-\log_{|R^{0}|^s}|C^{0}|^{s-1}|=n+1-\dfrac{s-1}{s}\log_{|R^{0}|}|C^{0}|=d(C^{0})+\dfrac{1}{s}\log_{|R^{0}|}|C^{0}|.$
According to the Singleton Bound, for a linear code $C^{\prime}$ of length $n$ and cardinality $|C|$ over $\mathbf{R}$,
 the minimum Hamming distance of $C^{\prime}$ satisfies $d(C^{\prime})\leq ~d(C^{0})+\dfrac{1}{s}\log_{|R^{0}|}|C^{0}|.$
So if $\log_{|R^{0}|}|C^{0}|< s$,  then $C$ has the maximal minimum Hamming distance among the linear codes of length $n$ and cardinality $|C|$ over $\mathbf{R}$.}
\end{remark}

\begin{example}\label{exa5}
{\rm Let $\mathbb{F}_{q}$ be the finite field of order $q$ and $\mathbf{R}=\psi^{-1}(\underbrace{\mathbb{F}_{q} \times \mathbb{F}_{q} \times \cdots \times \mathbb{F}_{q}}_{s}$), where $s\geq 2$.
Then $|\mathbf{R}|=q^s$.
Let $\alpha$ be a primitive element of $\mathbb{F}_{q}$ and $0< k < \min\{s,q\}$ be an integer.
Let $C_{0}$ be the cyclic code of length $q-1$ over $\mathbb{F}_{q}$ generated by $\prod_{i=0}^{q-1-k}(x-\alpha^i)$.
It means that $C_{0}$ is the $[q-1,k, q-k]$ Reed-Solomon code over $\mathbb{F}_{q}$, which is an MDS code.

Let $C=\text{CRT}(\underbrace{C_{0},\cdots,C_{0}}_{s-1}, 0)$ and $\lambda=\psi^{-1}(\underbrace{1,\ldots, 1}_{s-1}, 0)\in \mathbf{R}$.
Then $C$ is a $\lambda$-constacyclic code  of length $q-1$ over $\mathbf{R}$, $|C|=|C_0|^{s-1}=q^{k(s-1)}$ and $d(C)=d(C_0)=q-k$.
Let $C^{\prime}$ be a linear code of length $q-1$ and cardinality $q^{k(s-1)}$ over $\mathbf{R}$,
then by  the Singleton Bound, we have $d(C^{\prime}) \leq ~q-1-\log_{| \mathbf{R}|}q^{k(s-1)}+1=q-k+\dfrac{k}{s}.$
Since $0< \dfrac{k}{s} < 1 $, $d(C^{\prime}) \leq q-k$.
Thus $C$ is an optimal code over $\mathbf{R}$ in the sense that it achieves the maximum possible minimum Hamming distance for length $q-1$ and cardinality $q^{k(s-1)}$.}
\end{example}

\begin{example}\label{exa6}

{\rm Let $R^{0}$ be the Galois ring  of characteristic $p^t$ and cardinality $p^{tm}$ and $\mathbf{R}=\psi^{-1}(\underbrace{R^{0} \times R^{0} \times \cdots \times R^{0}}_{s}$), where $s\geq 2$.
Then $| \mathbf{R}|=p^{tms}$.
Let $n|(p^m-1)$ and $\alpha$ be an element of order $n$ in $R^{0}$.
Let $0< k < \min\{s,n\}$ and $C_{0}$ be the cyclic code of length $n$ over $R^{0}$ generated by $\prod_{i=0}^{n-k-1}(x-\alpha^i)$.
It is clear that $C_{0}$ is an MDS cyclic code which has minimum Hamming distance $n-k+1$ and cardinality $p^{tmk}$.

Let $C=\text{CRT} (\underbrace{C_{0},\cdots,C_{0}}_{s-1}, 0)$ and $\lambda=\psi^{-1}(\underbrace{1,\ldots, 1}_{s-1}, 0)\in \mathbf{R}$.
Then $C$ is a $\lambda$-constacyclic code  of length $n$ over $\mathbf{R}$, $|C|=|C_0|^{s-1}=p^{tmk(s-1)}$ and $d(C)=d(C_0)=n-k+1$.
Let $C^{\prime}$ be a linear code of length $n$ and cardinality $p^{tmk(s-1)}$ over $\mathbf{R}$,
then by the  Singleton Bound, we can get  $d(C^{\prime}) \leq ~n-\log_{| \mathbf{R}|}p^{tmk(s-1)}+1=n-k+1+\dfrac{k}{s}.$
Since $0< \dfrac{k}{s} < 1 $, we have $d(C^{\prime}) \leq n-k+1$.
Thus $C$ is an optimal code over $\mathbf{R}$ in the sense that it achieves the maximum possible minimum Hamming distance for length $n$ and cardinality $p^{tmk(s-1)}$.}
\end{example}

\bigskip
\textbf{Acknowledgments}~This work was supported by NSFC (Grant No. 11871025).

\end{document}